\documentclass[11pt]{article}
\usepackage{mathrsfs}
\usepackage{CJK}
\usepackage{bbm}
\usepackage{graphicx}
\usepackage{amsmath,amssymb}
\usepackage{amsfonts}
\usepackage{amsthm}
\usepackage{verbatim}
\usepackage{color}

\usepackage{float}

\usepackage[numbers,sort&compress]{natbib}

\newtheorem{Definition}{Definition}[section]
\newtheorem{Theorem}{Theorem}[section]
\newtheorem{Remark}{Remark}[section]
\newtheorem{Lemma}{Lemma}[section]

\newtheorem{Proposition}{Proposition}[section]

\numberwithin{equation}{section}

\begin{document}
\title{{\LARGE \textbf{Recovery of signals under the condition on RIC and ROC via prior support information }}}
\author{Wengu Chen$^{1}$ ,\ \ Yaling Li$^{2}$\\[5pt]
$^{1}$ Institute of Applied Physics and Computational Mathematics\\ Beijing, 100088, China\\
$^{2}$ Graduate School, China Academy of Engineering Physics\\ Beijing, 100088, China\\[5pt]
Email: chenwg@iapcm.ac.cn, leeyaling@126.com} \maketitle

\begin{bfseries}
Abstract
\end{bfseries}
In this paper, the sufficient condition in terms of the RIC and ROC
for the stable and  robust recovery of signals in both noiseless and noisy settings
was established
via weighted $l_{1}$
minimization when there is partial prior information on support of signals.
An improved performance guarantee has been derived.
 We can obtain a less restricted sufficient condition for
signal reconstruction and a tighter recovery error bound  under some conditions
via weighted $l_{1}$ minimization.
When  prior support
estimate is at least $50\%$ accurate, the sufficient condition is weaker than the analogous
 condition by standard $l_{1}$ minimization method,
meanwhile the reconstruction error upper bound is provably to be smaller under additional conditions.
Furthermore, the sufficient condition is also proved sharp.

\begin{bfseries}
Keywords
\end{bfseries}
Compressed sensing, Restricted isometry property, Restricted orthogonality constant, Weighted $l_{1}$
minimization, Sparse signal recovery

\begin{bfseries}
Mathematics Subject Classification (2010)
\end{bfseries}
94A12, 90C59, 94A15
\section{Introduction}
\label{intro} \ \ \ \ {Compressed}  sensing shows
that it is highly possible to reconstruct sparse signals  from what was previously
believed to be incomplete information\cite{CRT1,D}. The fundamental goal in
compressed sensing is to recover a high dimensional sparse signal based
on a small number of linear measurements, possibly corrupted by noise.
This can be compactly described via
\begin{align}\label{m1}
y=Ax+z,
\end{align}
where $A$ is a given $ n\times N$ sensing matrix with $n\ll N$,
i.e., using very few measurements, $y\in \mathbb{R}^{n}$ is a vector
of measurements, and $z \in \mathbb{R}^{n}$ is the measurement error
($z=0$ means no noise). One needs to reconstruct the unknown signal
$x\in\mathbb{R}^{N}$ based on $A$ and $y$. In general, the solutions
to the underdetermined systems of linear equations (\ref{m1}) are
not unique. In order to recover $x$ uniquely, additional assumptions
on $A$ such as restricted isometry property and $x$ such as sparsity
are needed. A vector $x\in\mathbb{R}^{N}$ is $k-$sparse if
$\|x\|_{0}=|supp(x)|\leq k$, where $supp(x)=\{i: x_{i}\neq 0\}$ is
the support of $x$. Then the most natural approach for solving this
problem is to find the sparsest solution in the feasible set of
possible solutions. In the noiseless case, it can be cast as the
$l_{0}$ minimization problem as below \cite{CRT1,D,LW,X}:
\begin{align}\label{f0}
  \underset{x\in \mathbb{R}^{N}}{\rm minimize}\quad \|x\|_{0} \quad {\rm subject\quad to} \quad Ax=y.
\end{align}
It was proved that when measurements $n>2k$ and $A$ is in general position
(any collection of $n$ columns of $A$ is linearly independent),
then any $k-$sparse signals can be exactly recovered \cite{DE}.
However, $l_{0}$ minimization problem is a combinatorial problem which
becomes intractable in the high dimensional settings.
Hence, solving it directly is NP-hard.



Cand\`es and Tao
\cite{CT} then proposed the following constrained $l_{1}$ minimization
method:
\begin{align}\label{f1}
  \underset{x\in \mathbb{R}^{N}}{\rm minimize}\quad\|x\|_{1} \ \ \ {\rm subject\quad to}\ \ \
  \|y-Ax\|_2\leq\epsilon.
\end{align}
It can be viewed as a convex relaxation of $l_{0}$ minimization.
To recover sparse signals via constrained $l_{1}$ minimization,
Cand\`es and Tao \cite{CT} also introduced the notion of Restricted
Isometry Property (RIP), which is one of the most commonly used
frameworks for compressive sensing. The definition of RIP is as
follows.
\begin{Definition}
Let $A\in \mathbb{R}^{n\times N}$ be a matrix and $1\leq k \leq N$
is an integer. The restricted isometry constant (RIC) $\delta_{k}$
of order $k$ is defined as the smallest nonnegative
 constant that
satisfies
$$(1-\delta_{k})\|x\|_{2}^{2}\leq\|Ax\|_{2}^{2}\leq(1+\delta_{k})\|x\|_{2}^{2},$$
for all $k-$sparse vectors $x\in\mathbb{R}^{N}.$ Note that for $k_{1}\leq k_{2}$,
$\delta_{k_{1}}\leq \delta_{k_{2}}$.
\end{Definition}

 Thus, $l_{1}$ minimization has been proved an effective way to recover sparse signals
 in many settings \cite{CZ2,CWX1,CWX,CZ1,ML,CT,CZ,CRT,CXZ}. Cand\`es, Romberg and Tao first gained the
 sufficient condition for stable recovery by $l_{1}$ minimization method \cite {CRT}.
In \cite{CZ}, Cai and Zhang applied the following $l_{1}$ minimization
 \begin{align}\label{f4}
  \underset{x\in \mathbb{R}^{N}}{\rm minimize}\quad\|x\|_{1} \ \ \ {\rm subject\quad to}\ \ \
  \|y-Ax\|_2\in \mathcal{B},
\end{align}
where $\mathcal{B}$ is a bounded set determined by the noise
structure. In particular, $\mathcal{B}$ is taken to be $ \{0\}$ in
the noiseless case. Here they considered the following $l_{2}$ bounded noise and Dantzing
 Selector noise settings
\begin{align}\label{b1}
 \mathcal{B}^{l_{2}}(\varepsilon)=\{z: \|z\|_{2}\leq\varepsilon\}
\end{align}
and
\begin{align}\label{b2}
 \mathcal{B}^{DS}(\varepsilon)=\{z:
 \|A^{T}z\|_{\infty}\leq\varepsilon\}.
\end{align}

Cai and Zhang \cite{CZ} provided a sharp sufficient condition $\delta_{tk}<\sqrt{\frac{t-1}{t}}$ with $t\geq 4/3$
which can guarantee the exact recovery of all $k-$sparse signals in the
noiseless case and stable recovery of  approximately sparse signals
in the noise case by $l_{1}$ minimization method (\ref{f4}) with (\ref{b1}) and (\ref{b2}).

In addition, the restricted orthogonality constant is also important in compressed sensing \cite{CWX,CWX1,CZ2}.
\begin{Definition}
Let $A\in \mathbb{R}^{n\times N}$ be a matrix and $1\leq k_{1}, k_{2}\leq N$ be integers
with $ k_{1}+k_{2}\leq N$, the restricted orthogonality constant (ROC) $\theta_{k_{1},k_{2}}$
of order $(k_{1},k_{2})$ is defined as the smallest nonnegative
constant that satisfies
$$|\langle Au, Av\rangle|\leq \theta_{k_{1},k_{2}}\|u\|_{2}\|v\|_{2},$$
for all $k_{1}-$sparse vectors $u\in\mathbb{R}^{N}$ and $k_{2}-$sparse vectors $v\in\mathbb{R}^{N}$
with disjoint supports.  Note that for $k_{1}\leq k_{2}$ and $k'_{1}\leq k'_{2}$,
$\theta_{k_{1},k'_{1}}\leq \theta_{k_{2},k'_{2}}$.
\end{Definition}
It also has been shown that $l_{1}$ minimization can recover a sparse signal under various conditions on $\delta_{k}$
 and $\theta_{k_{1},k_{2}}$  \cite{CXZ,CWX1,CT,CT1,CZ2,CWX,CWX2,DH,F}.
 For example, $\delta_{k}+\theta_{k,k}+\theta_{k,2k}<1$ \cite{CT},
 $\delta_{2k}+\theta_{k,2k}<1$ \cite{CT1}, $\delta_{1.5k}+\theta_{k,1.5k}<1$ \cite{CXZ} and
 $\delta_{1.25k}+\theta_{k,1.25k}<1$ \cite{CWX1}.
 Cai and Zhang \cite{CZ2} also established a sharp sufficient condition in terms of RIC and ROC to achieve the stable
and robust
 recovery of signals
 in both noiseless and noisy cases via $l_{1}$ minimization method.
In fact, Cai and Zhang  \cite{CZ2}
 proved that $\delta_{a}+C_{a,b,k}\theta_{a,b}<1$ can ensure stable
and robust recovery of signals via $l_{1}$ minimization method (\ref{f4}) with (\ref{b1}) and (\ref{b2}).
Moreover, for any $\varepsilon>0$, $\delta_{a}+C_{a,b,k}\theta_{a,b}<1+\varepsilon$
is not sufficient to guarantee the exact and stable recovery of all $k-$sparse signals via any methods.

It is worthy of noting that compressed sensing is a nonadaptive data acquisition
technique since $A$ is independent of $x$, the signal being measured.
The $l_{1}$ minimization method (\ref{f1}) is also itself
nonadaptive as a result of no prior information on the signal $x$ being used in (\ref{f4}).
In practical examples, however, the estimate of the support of the
signal or of its largest coefficients may be possible to be drawn.
 Incorporating prior information is very useful for recovering signals from compressive measurements.
Thus, the following weighted $l_{1}$ minimization method  which
incorporates partial support information of the signals has been
introduced to replace standard $l_{1}$ minimization
\begin{align}\label{f3}
 \underset{x\in \mathbb{R}^{N}}{\rm minimize}\quad\|x\|_{1,\mathrm{w}} \ \ \
 {\rm subject\quad to}\ \ \  \|y-Ax\|_2\leq\epsilon,
 \end{align}
where $\mathrm{w} \in [0, 1]^{N}$ and
$\|x\|_{1,\mathrm{w}}=\sum\limits_{i}\mathrm{w}_{i}|x_{i}|$.
Reconstructing compressively sampled signals with partially known
support has been previously studied in the literature; see
\cite{BMP,VL,LV,J,KXAH,LV1,FMSY}. Borries, Miosso and Potes in
\cite{BMP}, Khajehnejad $et~al.$ in \cite{KXAH}, and Vaswani and Lu
in \cite{VL} introduced the problem of signal recovery with
partially known support independently. The works by Borries $et~al.$
in \cite{BMP}, Vaswani and Lu in \cite{VL,LV,VL1} and Jacques in
\cite{J} incorporated known support information using weighted
$l_{1}$ minimization approach with zero weights on the known
support, namely, given a support estimate
$\widetilde{T}\subset\{1,2,\ldots, N\}$ of unknown signal $x$,
setting $\mathrm{w}_{i}=0$ whenever $i\in \widetilde{T}$ and
$\mathrm{w}_{i}=1$ otherwise, and derived sufficient recovery
conditions. Friedlander $et~al.$ in \cite{FMSY} extended weighted
$l_{1}$ minimization approach to nonzero weights.
They allow the
weights $\mathrm{w}_{i}=\omega\in [0, 1]$ if $i\in \widetilde{T}$.
Since Friedlander $et~al.$ incorporated the prior support
information and consider the accuracy of the support estimate, they
derived the stable and robust recovery guarantees for weighted
$l_{1}$ minimization which generalize the results of Cand\`es,
Romberg and Tao in \cite{CRT}.
They actually improved the recovery guarantees of $l_{1}$ minimization problem (\ref{f1})
by using weighted $l_{1}$ minimization problem (\ref{f3}).
  Friedlander $et~al.$
\cite{FMSY} pointed out that once at least $50\%$ of the support
information is accurate,
  a less conservative sufficient condition for guaranteeing stably and robustly signal reconstruction
  as well as a tighter reconstruction error bound can be obtained.
Furthermore, they also pointed out sufficient conditions are weaker
than those of \cite{VL} when $\omega=0$.



In this paper, we consider the following weighted $l_{1}$ minimization
 method:
\begin{align}\label{f2}
 \underset{x\in \mathbb{R}^{N}}{\rm minimize}\quad&\|x\|_{1,\mathrm{w}}\ \ \
 {\rm subject\quad to}
 \ \ \ y-Ax\in \mathcal{B}  \notag \\
 {\rm with }\ \ \ \mathrm{w}_{i}&=\left\{\begin{array}{cc}
                     1, & i\in \widetilde{T}^{c} \\
                     \omega, &  i\in \widetilde{T}.
                   \end{array}
                   \right.
 \end{align}
where $0\leq \omega \leq 1 $ and $\widetilde{T}\subset\{1,2,\ldots,
N\}$ is a given support estimate of unknown signal $x$. $\mathcal{B}$
is also a bounded set determined by the noise settings (\ref{b1})
and (\ref{b2}). Our goal is to generalize the results of Cai and
Zhang \cite{CZ2} via the weighted $l_{1}$ minimization method
(\ref{f2}). We establish the sufficient condition on RIC and ROC for the stable
and robust recovery of signals with partially known support
information from (\ref{m1}). We also show that the recovery by
weighted $l_{1}$ minimization method (\ref{f2}) is stable and robust
under weaker sufficient conditions compared to the standard $l_{1}$
minimization method (\ref{f4}) when we have the partial support
information with accuracy better than $50\%$. Meanwhile, we obtain
the smaller upper bounds on the reconstruction error under additional conditions.
By means of weighted $l_{1}$ minimization method (\ref{f2}),
 that is to say, the requirement on the RIC and ROC of the sensing matrix for guaranteeing
 stable and robust signal recovery can be further relaxed
 if at least $50\%$ of the support estimate is accurate; in addition,
 the reconstruction error upper bound is provably to be smaller under additional conditions.
Our result implies that the achievable performance of signal recovery via
 weighted $l_{1}$ minimization method (\ref{f2}) is actually better than the works
 by Cai and Zhang \cite{CZ2} under some conditions.

The rest of the paper is organized as follows. In Section \ref{2},
we will introduce some notations and some basic lemmas that will be
used. The main results  are given in Section \ref{3}, and the proofs
of our main results are presented in Section \ref{4}.

\section{Preliminaries}\label{2}
Let us begin with basic notations. For arbitrary
$x\in\mathbb{R}^{N}$, $x_{\max(k)}$ is defined as $x$ with all but the
largest $k$ entries in absolute value set to zero, i.e. $x_{\max(k)}$ is the
best $k-$term approximation of $x$,  and
$x_{-\max(k)}=x-x_{\max(k)}$. Let $T_{0}$ be the support of $x_{\max(k)}$,
with $T_{0}\subseteq
\{1,\ldots,N\}$ and $|T_{0}|\leq k$. Let $\widetilde{T}\subseteq
\{1, \ldots, N\}$ be the support estimate of $x$ with
$|\widetilde{T}|=\rho k$, where $\rho \geq 0$ represents the ratio
of the size of the estimated support to the size of the actual
support of $x_{\max(k)}$ (or the support of $x$ if $x$ is $k-$ sparse).
Denote $\widetilde{T}_{\alpha}=T_{0}\cap \widetilde{T}$ and
$\widetilde{T}_{\beta}=T_{0}^{c}\cap \widetilde{T}$ with
$|\widetilde{T}_{\alpha}|=\alpha |\widetilde{T}|=\alpha\rho k$ and
$|\widetilde{T}_{\beta}|=\beta |\widetilde{T}|=\beta\rho k$, where
$\alpha$ denotes the ratio of the number of indices in $T_0$ that
were accurately estimated in $\widetilde{T}$ to the size of
$\widetilde{T}$ and $\alpha+\beta=1$. For arbitrary nonnegative
number $\zeta$, we denote by $[[\zeta]]$ an integer satisfying $\zeta\leq
[[\zeta]] <\zeta+1.$ Moreover, for given set $T\subseteq
\{1,\ldots,N\}$, we denote by $x_{T}$ the vector which equals
to $x$ on $T$ and $0$ on the component $T^{c}$.

We first state three key technical tools used in the proof of the main result.
Lemma \ref{l1} was introduced by Cai and Zhang (\cite{CZ2}, Lemma 5.1)
which provides a way to estimate the inner product by the ROC when only one component
is sparse. Lemma \ref{l2} introduced by Cai and Zhang  (\cite{CZ1}, Lemma 5.3)
provides an inequality between the sum of the $\alpha$th
power of two sequences of nonnegative numbers based on the inequality of their sums.
Cai, Wang and Xu  (\cite{CWX1}, Lemma 1)
 supplied Lemma \ref{l3} that reveals the relationship between ROC's of different orders.
\begin{Lemma}[\cite{CZ2}, Lemma 5.1]\label{l1}
Let $k_{1},k_{2}\leq N$ and $\lambda\geq0$. Assume $u, v\in \mathbb{R}^{N}$
have disjoint supports and $u$ is $k_{1}-$sparse. If $\|v\|_{1}\leq\lambda k_{2}$
and $\|v\|_{\infty}\leq\lambda$, then
$$|\langle Au, Av\rangle|\leq\theta_{k_{1},k_{2}}\|u\|_{2}\cdot\lambda\sqrt{k_{2}}.$$
\end{Lemma}

\begin{Lemma}[\cite{CZ1}, Lemma 5.3]\label{l2}
Assume $m\geq k$, $a_{1}\geq a_{2}\geq\cdots\geq a_{m}\geq 0$,
$\sum\limits_{i=1}^{k}a_{i}\geq \sum\limits_{i=k+1}^{m}a_{i},$
then for all $\alpha\geq 1$,
$$\sum\limits_{j=k+1}^{m}a_{j}^{\alpha}\leq \sum\limits_{i=1}^{k}a_{i}^{\alpha}.$$
More generally, assume $a_{1}\geq a_{2}\geq\cdots\geq a_{m}\geq 0$, $\lambda\geq 0$
and $\sum\limits_{i=1}^{k}a_{i}+\lambda\geq \sum\limits_{i=k+1}^{m}a_{i},$ then for all $\alpha\geq 1$,
$$\sum\limits_{j=k+1}^{m}a_{j}^{\alpha}\leq k\Big(\sqrt[\alpha]{\frac{\sum_{i=1}^{k}a_{i}^{\alpha}}{k}}+\frac{\lambda}{k}\Big)^{\alpha}.$$
\end{Lemma}

\begin{Lemma}[\cite{CWX1}, Lemma 1]\label{l3}
For any $\tau\geq 1$ and positive integers $k,k'$ such that $\tau k'$ is an integer, then
$$\theta_{k,\tau k'}\leq \sqrt{\tau}\theta_{k,k'}.$$
\end{Lemma}
As we mentioned in the introduction, Cai and Zhang \cite{CZ2} provided the sharp sufficient condition for
 ensuring exact and stable sparse signals reconstruction via $l_{1}$ minimization (\ref{f4}).
 Their main result can be stated as below.

\begin{Theorem}[\cite{CZ2}, Theorem 2.6]\label{t1}
Let $y=Ax+z$ with $\|z\|_{2}\leq \varepsilon$ and
$\widehat{x}^{l_{2}}$ is the minimizer of {\rm(\ref{f4})} with
$\mathcal{B}=\mathcal{B}^{l_{2}}(\eta)=\{z:\|z\|_{2}\leq \eta\}$ for
some $\eta \geq \varepsilon$. If
\begin{align}\label{g1}
 \delta_{a}+C_{a,b,k}\theta_{a,b}<1
\end{align}
for some positive integers $a$ and $b$ with $1\leq a\leq k$, where
\begin{align}\label{g4}
   C_{a,b,k}=\max \left\{ \frac{2k-a}{\sqrt{ab}}, \sqrt{\frac{2k-a}{a}}\right\},
\end{align}
then
\begin{align}\label{g2}
  \|\widehat{x}^{l_{2}}-x\|_{2}\leq C_{0}(\varepsilon+\eta)
  +C_{1}\cdot2\|x_{-\max{(k)}}\|_{1},
\end{align}
where
\begin{align}\label{g21}
  C_{0}=\frac{\sqrt{2(1+\delta_{a})k/a}}{1-\delta_{a}-C_{a,b,k}\theta_{a,b}},\ \ \
  C_{1}=\frac{\sqrt{2k}C_{a,b,k}\theta_{a,b}}{(1-\delta_{a}-C_{a,b,k}\theta_{a,b})(2k-a)}+\frac{1}{\sqrt{k}}.
\end{align}
\end{Theorem}

\begin{Theorem}[\cite{CZ2}, Theorem 2.7]\label{t2}
Let $y=Ax+z$ with $\|A^{T}z\|_{\infty}\leq \varepsilon$ and
$\widehat{x}^{DS}$ is the minimizer of {\rm(\ref{f4})} with
$\mathcal{B}=\mathcal{B}^{DS}(\eta)=\{z:\|A^{T}z\|_{\infty}\leq
\eta\}$ for some $\eta \geq \varepsilon$. If
$\delta_{a}+C_{a,b,k}\theta_{a,b}<1$
 for some positive integers $a$ and $b$ with $1\leq a\leq k$, where
$ C_{a,b,k}=\max \left\{ \frac{2k-a}{\sqrt{ab}}, \sqrt{\frac{2k-a}{a}}\right\},$
 then
\begin{align}\label{g3}
  \|\widehat{x}^{DS}-x\|_{2}\leq  C'_{0}(\varepsilon+\eta)
  + C'_{1}\cdot2\|x_{-\max{(k)}}\|_{1},
\end{align}
where
\begin{align}\label{g22}
  C'_{0}=\frac{\sqrt{2k}}{1-\delta_{a}-C_{a,b,k}\theta_{a,b}},\ \ \ C'_{1}=C_{1}.
\end{align}
\end{Theorem}

Cai and Zhang pointed out that the sufficient condition (\ref{g1})
is sharp in Theorem 2.8 (see \cite{CZ2}). Namely, if
$\delta_{a}+C_{a,b,k}\theta_{a,b}=1$, there does not exist any
method that can exactly recover all $k-$sparse signals in noiseless
case. Also, in noisy case, for any $\varepsilon>0$,~
$\delta_{a}+C_{a,b,k}\theta_{a,b}<1+\varepsilon$ can not guarantee
the stable recovery of all $k-$sparse signals.

\section{Main results}\label{3}
\begin{Theorem}\label{t3}
Let $x\in\mathbb{R}^{N}$ be an arbitrary signal and its best $k-$term approximation support on $T_{0}\subseteq \{1, \ldots, N\}$ with $|T_{0}|\leq k$.
Let $\widetilde{T}\subseteq \{1, \ldots, N\}$ be an arbitrary set and denote $\rho \geq 0$ and $0\leq \alpha \leq 1$ such that $|\widetilde{T}|=\rho k$ and $|\widetilde{T}\cap T_{0}|=\alpha \rho k.$
 Let $y=Ax+z$ with $\|z\|_{2}\leq\varepsilon$ and $\widehat{x}^{l_{2}}$ is the minimizer of (\ref{f2}) with (\ref{b1}). If
\begin{align}\label{g23}
  \delta_{a}+C_{a,b,k}^{\alpha,\omega}\theta_{a,b}<1
\end{align}
 for some positive integers $a$ and $b$ with $1\leq a\leq k$, where
\begin{align}\label{g5}
 C_{a,b,k}^{\alpha,\omega}=\max\left\{ \frac{s}{\sqrt{ab}}, \sqrt{\frac{s}{a}}\right\}
\end{align}
 with
\begin{align}\label{g6}
 s=\left[\left[k-a+\omega k+(1-\omega)\sqrt{(1+\rho-2\alpha \rho)k}\cdot\max\{\sqrt{(1+\rho-2\alpha \rho)k}, \sqrt{a}\}\right]\right].
 \end{align}
Then
\begin{align}\label{g9}
\|\widehat{x}^{l_{2}}-x\|_{2}
\leq D_{0}(2\varepsilon)
+D_{1}\cdot2\left(\omega\|x_{T_{0}^{c}}\|_{1}+(1-\omega)\|x_{\widetilde{T}^{c}\cap T_{0}^{c}}\|_{1}\right),
 \end{align}
where
\begin{equation}\label{g24}
  \begin{split}
     D_{0}&= \frac{\sqrt{2(1+\delta_{a})d/a}}{1-\delta_{a}-C_{a,b,k}^{\alpha,\omega}\theta_{a,b}},\\
     D_{1}& =\frac{\sqrt{2d}C_{a,b,s}^{\alpha,\omega}\theta_{a,b}}{(1-\delta_{a}-C_{a,b,k}^{\alpha,\omega}\theta_{a,b})s}+\frac{1}{\sqrt{d}}.
\end{split}
\end{equation}

  Let $y=Ax+z$ with $\|A^{T}z\|_{\infty}\leq\varepsilon$. Assume that $\widehat{x}^{DS}$ is the minimizer of (\ref{f2}) with (\ref{b2}) and (\ref{g23}) holds. If
$$\delta_{a}+C_{a,b,k}^{\alpha,\omega}\theta_{a,b}<1$$
for some positive integers $a$ and $b$ with $1\leq a\leq k$, where
$$
 C_{a,b,k}^{\alpha,\omega}=\max\left\{ \frac{s}{\sqrt{ab}}, \sqrt{\frac{s}{a}}\right\},
 $$ where $s$ is given in (\ref{g6}).
Then
 \begin{align}\label{g10}
\|\widehat{x}^{DS}-x\|_{2}
\leq  D'_{0}(2\varepsilon)+ D'_{1}
\cdot2\left(\omega\|x_{T_{0}^{c}}\|_{1}+(1-\omega)\|x_{\widetilde{T}^{c}\cap T_{0}^{c}}\|_{1}\right),
\end{align}
where
\begin{equation}\label{g25}
  \begin{split}
    D'_{0}=\frac{\sqrt{2d}}{1-\delta_{a}-C_{a,b,k}^{\alpha,\omega}\theta_{a,b}},\ \
    D'_{1}=D_{1}.
   \end{split}
\end{equation}
Here \begin{align}\label{g7}
      d=\left\{
      \begin{array}{cc}
        k, &\omega=1, \\
        \max\{k, ~(1+\rho-2\alpha\rho)k\}, &0\leq \omega <1 .
      \end{array}
      \right.
     \end{align}
\end{Theorem}

\begin{Remark}
In Theorem \ref{t3}, we observed that every signal $x\in\mathbb{R}^{N}$ can be stably and robustly recovered.
And if $\mathcal{B}=\{0\}$ and $x$ is a $k-$sparse signal, then Theorem \ref{t3} ensures exact recovery of the signal $x$.
\end{Remark}

When the the measurement model (\ref{m1}) is with Gaussian noise,
the above results on the bounded noise case can be directly applicable to the case where the noise is Gaussian
by using the same argument as in \cite{CXZ,CWX1}.
This is due to the fact Gaussian noise is essentially bounded.
The concrete content is stated as follows.
\begin{Remark}
Let $x\in\mathbb{R}^{N}$ be an arbitrary signal and its best $k-$term approximation support on $T_{0}\subseteq \{1, \ldots, N\}$ with $|T_{0}|\leq k$.
Let $\widetilde{T}\subseteq \{1, \ldots, N\}$ be an arbitrary set and define $\rho \geq 0$ and $0\leq \alpha \leq 1$ such that $|\widetilde{T}|=\rho k$ and $|\widetilde{T}\cap T_{0}|=\alpha \rho k.$
Assume that $z\sim N_{n}(0, \sigma^{2}I)$ in (\ref{m1}) and
$\delta_{a}+C_{a,b,k}^{\alpha,\omega}\theta_{a,b}<1$ for some positive integers $a$ and $b$ with $1\leq a\leq k$, where
 $C_{a,b,k}^{\alpha,\omega}=\max\left\{ \frac{s}{\sqrt{ab}}, \sqrt{\frac{s}{a}}\right\}$
 with $s=\left[\left[k-a+\omega k+(1-\omega)\sqrt{(1+\rho-2\alpha \rho)k}\cdot\max\{\sqrt{(1+\rho-2\alpha \rho)k}, \sqrt{a}\}\right]\right]$. Let
$\mathcal{B}^{l_{2}}=\{z: \|z\|_{2}\leq\sigma\sqrt{n+2\sqrt{n\log n}}\}$
and $\mathcal{B}^{DS}=\{z: \|A^{T}z\|_{\infty}\leq\sigma\sqrt{2\log N}\}$.
$\widehat{x}^{l_{2}}$ and $\widehat{x}^{DS}$ is the minimizer of (\ref{f2}) with $\mathcal{B}^{l_{2}}$ and $\mathcal{B}^{DS}$, respectively.
Then, with probability at least $1-1/n$,
\begin{align*}
  \|\widehat{x}^{l_{2}}-x\|_{2}
&\leq D_{0}(2\sigma\sqrt{n+2\sqrt{n\log n}})+D_{1}
\cdot2\left(\omega\|x_{T_{0}^{c}}\|_{1}+(1-\omega)\|x_{\widetilde{T}^{c}\cap T_{0}^{c}}\|_{1}\right),
\end{align*}
and
\begin{align*}
  \|\widehat{x}^{DS}-x\|_{2}
\leq D'_{0}(2\sigma\sqrt{2\log N})+D'_{1}
\cdot2\left(\omega\|x_{T_{0}^{c}}\|_{1}+(1-\omega)\|x_{\widetilde{T}^{c}\cap T_{0}^{c}}\|_{1}\right),
\end{align*}
 with probability at least $1-1/\sqrt{\pi\log N}$.
 \end{Remark}

\begin{Theorem}\label{t4}
Let $1\leq a\leq s \leq k$, $a+s\leq N$ and $b\geq 1$, where $s$ is defined as (\ref{g6}).
Then there exists a sensing matrix $A\in \mathbb{R}^{n\times N}$ satisfying $\delta_{a}+C_{a,b,k}^{\alpha,\omega}\theta_{a,b}=1$
where
 $C_{a,b,k}^{\alpha,\omega}=\max\left\{ \frac{s}{\sqrt{ab}}, \sqrt{\frac{s}{a}}\right\}$
 and some $k-$sparse vector $\eta \in\mathbb{R}^{N}$ such that the weighted $l_{1}$ minimization method (\ref{f2})
 fails to exactly recover the $k-$sparse vector $\eta$ in the noiseless case
 and stably recover the $k-$sparse vector $\eta$ in the noise case.
 \end{Theorem}

\begin{Remark}
Theorem \ref{t4} implies that for arbitrarily $\varepsilon>0$, $\delta_{a}+C_{a,b,k}^{\alpha,\omega}\theta_{a,b}<1+\varepsilon$
is not sufficient to guarantee the exact recovery of all $k-$sparse vectors in noiseless case
and the stable recovery of all $k-$sparse vectors in noise case.
\end{Remark}

\begin{Proposition}\label{p1}
Let $s$ be defined as (\ref{g6}) and $d$ be defined as (\ref{g7}).
\begin{description}
  \item[{\rm (1)}] If $\omega=1$, then $s=2k-a, d=k$. The sufficient condition (\ref{g23}) of Theorem \ref{t3} is identical to that of Theorem \ref{t1} and Theorem \ref{t2} with (\ref{g1}), and $D_{0}=C_{0}, D_{1}=C_{1}, D'_{0}=C'_{0}, D'_{1}=C'_{1}$. Moreover, the condition is sharp.
  \item[{\rm (2)}]If $\alpha=\frac{1}{2}$, then $s=2k-a$ and $d=k$. The sufficient condition (\ref{g23}) of Theorem \ref{t3} is identical to that of Theorem \ref{t1} and Theorem \ref{t2} with (\ref{g1}), and $D_{0}=C_{0}, D_{1}=C_{1}, D'_{0}=C'_{0}, D'_{1}=C'_{1}$. Moreover, the condition is sharp.
  \item[{\rm (3)}] Assume $0\leq\omega<1$. If $\alpha>\frac{1}{2}$, then $s<2k-a$ and $d=k$. The sufficient condition (\ref{g23}) in Theorem \ref{t3}
  is weaker than that of  Theorem \ref{t1} and Theorem \ref{t2} with (\ref{g1}), and $D_{0}<C_{0}, D'_{0}<C'_{0}$.
\item[{\rm (4)}]Suppose $0\leq\omega<1$. If $\alpha>\frac{1}{2}$ and $b\leq s$, then $D_{1}<C_{1}$.
\item[{\rm (5)}]Suppose $0\leq\omega<1$. If $\alpha>\frac{1}{2}$ and $s<b\leq 2k-a$, then $D_{1}<C_{1}$
if and only if $1-\delta_{a}-C_{a,b,k}^{\alpha,\omega}\theta_{a,b}<\frac{2k-a-\sqrt{bs}}{\sqrt{a}(\sqrt{b}-\sqrt{s})}\theta_{a,b}$.
\item[{\rm (6)}]Suppose $0\leq\omega<1$. If $\alpha>\frac{1}{2}$ and $b>2k-a$, then $D_{1}<C_{1}$
if and only if
$1-\delta_{a}-C_{a,b,k}^{\alpha,\omega}\theta_{a,b}
<\sqrt{\frac{2k-a}{a}}\theta_{a,b}$.

\end{description}
\end{Proposition}

\section{Proofs}\label{4}
\begin{proof}[Proof of Theorem \ref{t3}]
Firstly, we show the estimate (\ref{g9}). Let $h=\widehat{x}^{l_{2}}-x$,
where $x$ is the original signal and $\widehat{x}^{l_{2}}$ is the minimizer of (\ref{f2}) with (\ref{b1}).
We can express $h$ as $h=\sum\limits_{i=1}^{N}c_{i}u_{i}$, where $\{c_{i}\}_{i=1}^{N}$ are nonnegative and decreasing, i.e.
$c_{1}\geq c_{2} \geq\cdots \geq c_{N}\geq0,$
$\{u_{i}\}_{i=1}^{N}$ are different unit vectors with one entry of $\pm 1$ and other
entries of zeros.
From the following inequality proved by Friedlander $et~ al.$ (see (21) in \cite{FMSY})
\begin{align}\label{g11}
  \|h_{T_{0}^{c}}\|_{1}\leq\omega\|h_{T_{0}}\|_{1}+(1-\omega)\|h_{T_{0}\cup\widetilde{T}\setminus\widetilde{T}_{\alpha}}\|_{1}
  +2\left(\omega\|x_{T_{0}^{c}}\|_{1}+(1-\omega)\|x_{\widetilde{T}^{c}\cap T_{0}^{c}}\|_{1}\right),
\end{align}
we have
\begin{align*}
  \sum\limits_{i=k+1}^{N}c_{i}=\|h_{-\max(k)}\|_{1}\leq\omega\|h_{T_{0}}\|_{1}+(1-\omega)\|h_{T_{0}\cup\widetilde{T}\setminus\widetilde{T}_{\alpha}}\|_{1}
  +2\left(\omega\|x_{T_{0}^{c}}\|_{1}+(1-\omega)\|x_{\widetilde{T}^{c}\cap T_{0}^{c}}\|_{1}\right).
\end{align*}
Noting that $|T_{0}\cup\widetilde{T}\setminus\widetilde{T}_{\alpha}|=(1+\rho-2\alpha\rho)k$, thus
\begin{equation*}
   \|h_{-\max(a)}\|_{\infty}=c_{a+1}\leq \frac{\sum\limits_{i=1}^{a}c_{i}}{a}
   =\frac{\|h_{\max(a)}\|_{1}}{a}\leq\frac{\|h_{\max(a)}\|_{2}}{\sqrt{a}},
\end{equation*}

\begin{align*}
  \|h_{-\max(a)}\|_{1}=&\sum\limits_{i=a+1}^{k}c_{i}+\sum\limits_{i=k+1}^{N}c_{i} \\
\leq& \frac{k-a}{k}\sum\limits_{i=1}^{k}c_{i}+\omega\|h_{T_{0}}\|_{1}+(1-\omega)\|h_{T_{0}\cup\widetilde{T}\setminus\widetilde{T}_{\alpha}}\|_{1}
  +2\left(\omega\|x_{T_{0}^{c}}\|_{1}+(1-\omega)\|x_{\widetilde{T}^{c}\cap T_{0}^{c}}\|_{1}\right)  \\
\leq&  \frac{k-a}{a}\|h_{\max(a)}\|_{1}+\omega\sqrt{k}\|h_{T_{0}}\|_{2}
  +(1-\omega)\sqrt{(1+\rho-2\alpha\rho)k}\|h_{T_{0}\cup\widetilde{T}\setminus\widetilde{T}_{\alpha}}\|_{2}\\
  &+2\left(\omega\|x_{T_{0}^{c}}\|_{1}+(1-\omega)\|x_{\widetilde{T}^{c}\cap T_{0}^{c}}\|_{1}\right)  \\
 \leq&\frac{k-a}{\sqrt{a}}\|h_{\max(a)}\|_{2}+\omega\sqrt{k}\sqrt{\frac{k}{a}}\|h_{\max(a)}\|_{2}  \\
  &+(1-\omega)\sqrt{(1+\rho-2\alpha\rho)k}\cdot\max\left\{\sqrt{\frac{(1+\rho-2\alpha\rho)k}{a}}, 1\right\}\|h_{\max(a)}\|_{2}\\
  &+2\left(\omega\|x_{T_{0}^{c}}\|_{1}+(1-\omega)\|x_{\widetilde{T}^{c}\cap T_{0}^{c}}\|_{1}\right)    \\
  =&\left(k-a+\omega k+(1-\omega)\sqrt{(1+\rho-2\alpha \rho)k}\cdot\max\{\sqrt{(1+\rho-2\alpha \rho)k}, \sqrt{a}\}\right)\frac{\|h_{\max(a)}\|_{2}}{\sqrt{a}} \\
 &+2\left(\omega\|x_{T_{0}^{c}}\|_{1}+(1-\omega)\|x_{\widetilde{T}^{c}\cap T_{0}^{c}}\|_{1}\right)  \\
 \leq&s\frac{\|h_{\max(a)}\|_{2}}{\sqrt{a}}+2\left(\omega\|x_{T_{0}^{c}}\|_{1}+(1-\omega)\|x_{\widetilde{T}^{c}\cap T_{0}^{c}}\|_{1}\right),
  \end{align*}
where $s=\left[\left[k-a+\omega k+(1-\omega)\sqrt{(1+\rho-2\alpha \rho)k}\cdot\max\{\sqrt{(1+\rho-2\alpha \rho)k}, \sqrt{a}\}\right]\right].$
Taking $k_{1}=a, k_{2}=s, \lambda=\frac{\|h_{\max(a)}\|_{2}}{\sqrt{a}}+\frac{2\left(\omega\|x_{T_{0}^{c}}\|_{1}+(1-\omega)\|x_{\widetilde{T}^{c}\cap T_{0}^{c}}\|_{1}\right)}{s}$, from above inequalities and Lemma \ref{l1} , we obtain
\begin{align*}
  |\langle Ah_{\max(a)}, Ah_{-\max(a)}\rangle|\leq&\theta_{a,s}\|h_{\max(a)}\|_{2}\sqrt{s}
  \cdot\left(\frac{\|h_{\max(a)}\|_{2}}{\sqrt{a}}
  +\frac{2\left(\omega\|x_{T_{0}^{c}}\|_{1}+(1-\omega)\|x_{\widetilde{T}^{c}\cap T_{0}^{c}}\|_{1}\right)}{s}\right).
  \end{align*}
Combining the definition of $\delta_{k}$ and the fact that
\begin{align}\label{g14}
  \|Ah\|_{2}= \|A\widehat{x}^{l_{2}}-Ax\|_{2}\leq\|y-A\widehat{x}^{l_{2}}\|_{2}+\|Ax-y\|_{2}\leq 2\varepsilon,
\end{align}
 we have
\begin{align}\label{g15}
 |\langle Ah_{\max(a)}, Ah\rangle|
  &\leq \|Ah_{\max(a)}\|_{2}\|Ah\|_{2} \notag\\
  &\leq \sqrt{1+\delta_{a}}\|h_{\max(a)}\|_{2}\cdot(2\varepsilon).
\end{align}
Hence,
\begin{align*}
  (2\varepsilon)&\sqrt{1+\delta_{a}}\|h_{\max(a)}\|_{2}\geq |\langle Ah_{\max(a)}, Ah\rangle| \\
   & \geq \|Ah_{\max(a)}\|_{2}^{2}-|\langle Ah_{\max(a)}, Ah_{-\max(a)}\rangle|  \\
   &\geq (1-\delta_{a})\|h_{\max(a)}\|_{2}^{2}-\theta_{a,s}\|h_{\max(a)}\|_{2}\sqrt{s}
  \cdot\left(\frac{\|h_{\max(a)}\|_{2}}{\sqrt{a}}
  +\frac{2\left(\omega\|x_{T_{0}^{c}}\|_{1}+(1-\omega)\|x_{\widetilde{T}^{c}\cap T_{0}^{c}}\|_{1}\right)}{s}\right)\\
  &=\left(1-\delta_{a}-\sqrt{\frac{s}{a}}\theta_{a,s}\right)\|h_{\max(a)}\|_{2}^{2}
  -\theta_{a,s}\|h_{\max(a)}\|_{2}\frac{2\left(\omega\|x_{T_{0}^{c}}\|_{1}+(1-\omega)\|x_{\widetilde{T}^{c}\cap T_{0}^{c}}\|_{1}\right)}{\sqrt{s}}.
\end{align*}

It follows from the above inequality that
\begin{align*}
  \|h_{\max(a)}\|_{2}&\leq \frac{\sqrt{1+\delta_{a}}(2\varepsilon)}{1-\delta_{a}-\sqrt{\frac{s}{a}}\theta_{a,s}}      +\frac{\theta_{a,s}}{1-\delta_{a}-\sqrt{\frac{s}{a}}\theta_{a,s}}
  \frac{2\left(\omega\|x_{T_{0}^{c}}\|_{1}+(1-\omega)\|x_{\widetilde{T}^{c}\cap T_{0}^{c}}\|_{1}\right)}{\sqrt{s}}. \\
\end{align*}

Define
\begin{align*}
  d=\left\{
     \begin{array}{cc}
       k, & \omega=1, \\
       \max\{ k, (1+\rho-2\alpha\rho)k\}, & 0\leq \omega<1.
     \end{array}
 \right.
\end{align*}

With (\ref{g11}), it is clear that
\begin{align*}
  \|h_{-\max(d)}\|_{1}\leq \|h_{\max(d)}\|_{1}+2\left(\omega\|x_{T_{0}^{c}}\|_{1}+(1-\omega)\|x_{\widetilde{T}^{c}\cap T_{0}^{c}}\|_{1}\right).
\end{align*}
From Lemma \ref{l2}, we have
\begin{align*}
  \|h_{-\max(d)}\|_{2}\leq \|h_{\max(d)}\|_{2}+\frac{2\left(\omega\|x_{T_{0}^{c}}\|_{1}+(1-\omega)\|x_{\widetilde{T}^{c}\cap T_{0}^{c}}\|_{1}\right)}{\sqrt{d}}.
\end{align*}
 Therefore,
\begin{align*}
\|h\|_{2}&=\sqrt{\|h_{\max(d)}\|_{2}^{2}+\|h_{-\max(d)}\|_{2}^{2}}  \\
&\leq\sqrt{\|h_{\max(d)}\|_{2}^{2}+\left(\|h_{\max(d)}\|_{2}
+\frac{2\left(\omega\|x_{T_{0}^{c}}\|_{1}+(1-\omega)\|x_{\widetilde{T}^{c}\cap T_{0}^{c}}\|_{1}\right)}{\sqrt{d}}\right)^{2}} \\
&\leq\sqrt{2\|h_{\max(d)}\|_{2}^{2}}+\frac{2\left(\omega\|x_{T_{0}^{c}}\|_{1}+(1-\omega)\|x_{\widetilde{T}^{c}\cap T_{0}^{c}}\|_{1}\right)}{\sqrt{d}} \\
&=\sqrt{2\sum\limits_{i=1}^{d}c_{i}^{2}}+\frac{2\left(\omega\|x_{T_{0}^{c}}\|_{1}+(1-\omega)\|x_{\widetilde{T}^{c}\cap T_{0}^{c}}\|_{1}\right)}{\sqrt{d}} \\
&\leq\sqrt{\frac{2d}{a}\sum\limits_{i=1}^{a}c_{i}^{2}}+\frac{2\left(\omega\|x_{T_{0}^{c}}\|_{1}+(1-\omega)\|x_{\widetilde{T}^{c}\cap T_{0}^{c}}\|_{1}\right)}{\sqrt{d}} \\
&=\sqrt{\frac{2d}{a}}\|h_{\max(a)}\|_{2}+\frac{2\left(\omega\|x_{T_{0}^{c}}\|_{1}+(1-\omega)\|x_{\widetilde{T}^{c}\cap T_{0}^{c}}\|_{1}\right)}{\sqrt{d}} \\
&\leq\frac{\sqrt{2(1+\delta_{a})d/a}}{1-\delta_{a}-\sqrt{\frac{s}{a}}\theta_{a,s}}(2\varepsilon)
+\left(\frac{\sqrt{2d/a}\theta_{a,s}}{(1-\delta_{a}-\sqrt{\frac{s}{a}}\theta_{a,s})\sqrt{s}}+\frac{1}{\sqrt{d}}\right)
  \cdot2\left(\omega\|x_{T_{0}^{c}}\|_{1}+(1-\omega)\|x_{\widetilde{T}^{c}\cap T_{0}^{c}}\|_{1}\right).
\end{align*}

Since
\begin{align*}
  \theta_{a,s}=\theta_{a,\frac{s}{\min\{b,s\}}\min\{b,s\}}
  \leq \sqrt{\frac{s}{\min\{b,s\}}}\theta_{a,\min\{b,s\}}
   \leq \max\{\sqrt{\frac{s}{b}},1\}\theta_{a,b}
   = \sqrt{\frac{a}{s}}C_{a,b,k}^{\alpha,\omega}\theta_{a,b},
\end{align*}
where $C_{a,b,k}^{\alpha,\omega}=\max\{\frac{s}{\sqrt{ab}},\sqrt{\frac{s}{a}}\},$  and the first inequality follows from Lemma \ref{l3}.
Consequently,
\begin{align*}
\|h\|_{2}
&\leq\frac{\sqrt{2(1+\delta_{a})d/a}}{1-\delta_{a}-C_{a,b,k}^{\alpha,\omega}\theta_{a,b}}(2\varepsilon)
+\left(\frac{\sqrt{2d}C_{a,b,k}^{\alpha,\omega}\theta_{a,b}}{(1-\delta_{a}-C_{a,b,k}^{\alpha,\omega}\theta_{a,b})s}+\frac{1}{\sqrt{d}}\right)
  \cdot2\left(\omega\|x_{T_{0}^{c}}\|_{1}+(1-\omega)\|x_{\widetilde{T}^{c}\cap T_{0}^{c}}\|_{1}\right)
\end{align*}
So, (\ref{g9}) is obtained.

Next, we can prove (\ref{g10}) going along similar lines to that of (\ref{g9}). To prove(\ref{g10}),
we only need to use the following (\ref{g19}) and (\ref{g20}) instead of (\ref{g14}) and (\ref{g15}), respectively.

\begin{align}\label{g19}
  \|A^{T}Ah\|_{\infty} & = \|A^{T}A(\widehat{x}^{DS}-x)\|_{\infty} \notag\\
   & \leq \|A^{T}(A\widehat{x}^{DS}-y)\|_{\infty}+\|A^{T}(y-Ax)\|_{\infty} \notag\\
   &\leq 2\varepsilon,
\end{align}
\begin{align}\label{g20}
  |\langle Ah_{\max(a)}, Ah\rangle |& =|\langle h_{\max(a)}, ~A^{T}Ah\rangle| \notag\\
   & \leq \|h_{\max(a)}\|_{1}\|A^{T}Ah\|_{\infty} \notag\\
   &\leq \sqrt{a}\|h_{\max(a)}\|_{2}\cdot(2\varepsilon).
\end{align}
This completes the proof of Theorem \ref{t3}.
\end{proof}

\begin{proof}[Proof of Theorem \ref{t4}]
Firstly, let $L=a+s$, and
\begin{align*}
  \xi_{1}&=\frac{1}{\sqrt{L}}(\overbrace{1,\ldots,1}^{L},0,\ldots,0)\in\mathbb{R}^{N}, \quad\mathrm{if}~ L-k>\rho k, \\
  \mathrm{or}\quad
  \xi_{1}&=\frac{1}{\sqrt{L}}(\underbrace{1,\ldots, 1}_{k-\alpha\rho k}, \underbrace{\overbrace{1,\ldots,1}^{L-k},0,\ldots,0}_{\rho k},
  \underbrace{1,\ldots, 1}_{\alpha\rho k}, 0,\ldots,0)\in\mathbb{R}^{N}, \quad\mathrm{if} ~L-k\leq\rho k,
\end{align*}
Due to $\|\xi_{1}\|_{2}=1$, we extend $\xi_{1}$ into an orthonormal basis $\{\xi_{1},\ldots,\xi_{N}\}$ of $\mathbb{R}^{N}$.
Next, we define the linear map $A: \mathbb{R}^{N}\rightarrow\mathbb{R}^{N}$ such that for all $x=\sum\limits_{i=1}^{N}c_{i}\xi_{i}\in\mathbb{R}^{N},$
\begin{align*}
 Ax&=\sqrt{1+\frac{L-s}{L+s}}(x-\langle \xi_{1}, x\rangle \xi_{1})=\sqrt{1+\frac{L-s}{L+s}}\sum\limits_{i=2}^{N}c_{i}\xi_{i}.
 \end{align*}

Then for any $a-$sparse signal $x$, we can easily gain
$$\|Ax\|_{2}^{2} =\left(1+\frac{L-s}{L+s}\right)\left(\|x\|_{2}^{2}-|\langle \xi_{1}, x\rangle|^{2}\right),$$
and
$$|\langle \xi_{1}, x \rangle|^{2}\leq \|x\|_{2}^{2}\cdot\sum\limits_{i\in \mathrm{supp}(x)}|\xi_{1}(i)|^{2}
\leq \|x\|_{2}^{2}\cdot\|\xi_{1,\max(a)}\|_{2}^{2}\leq \frac{a}{L}\|x\|_{2}^{2}=\frac{L-s}{L}\|x\|_{2}^{2}.$$

Hence,
\begin{align*}
      \left(1+\frac{L-s}{L+s}\right)\|x\|_{2}^{2}\geq\|Ax\|_{2}^{2}
      \geq\left(1+\frac{L-s}{L+s}\right)(1-\frac{L-s}{L})\|x\|_{2}^{2}=\left(1-\frac{L-s}{L+s}\right)\|x\|_{2}^{2},
     \end{align*}
which deduces $$\delta_{a}\leq \frac{L-s}{L+s}.$$
Finally, we estimate $\theta_{a,b}.$ For arbitrary $a-$sparse vector $u\in\mathbb{R}^{N}$
and $b-$sparse vector $v\in\mathbb{R}^{N}$ with disjoint supports, we define
$u=\sum_{i=1}^{N}l_{i}\xi_{i}$ and  $v=\sum_{i=1}^{N}d_{i}\xi_{i}$. It follows immediately that
$0=\langle u, v\rangle=\sum_{i=1}^{N}l_{i}d_{i}$.

$\mathrm{(i)}$ When $b\leq s,$  through a simple calculation, it can be concluded that
$$|l_{1}|=|\langle \xi_{1}, u\rangle|\leq\|u\|_{2}\cdot\left(\sum\limits_{i\in \mathrm{supp}(u)}|\xi_{1}(i)|^{2}\right)^{1/2}
\leq \|u\|_{2}\cdot\|\xi_{1,\max(a)}\|_{2}\leq \sqrt{\frac{a}{L}}\|u\|_{2},$$
and
$$|d_{1}|=|\langle \xi_{1}, v\rangle|\leq\|v\|_{2}\cdot\left(\sum\limits_{i\in \mathrm{supp}(v)}|\xi_{1}(i)|^{2}\right)^{1/2}
\leq \|v\|_{2}\cdot\|\xi_{1,\max(b)}\|_{2}\leq \sqrt{\frac{b}{L}}\|v\|_{2}.$$
It then follows that
$$\frac{1}{1+\frac{L-s}{L+s}}|\langle Au, Av\rangle|=|\sum\limits_{i=2}^{N}l_{i}d_{i}|=|-l_{1}d_{1}|\leq \frac{\sqrt{ab}}{L}\|u\|_{2}\|v\|_{2}.$$
Accordingly,
$$\theta_{a,b}\leq (1+\frac{L-s}{L+s})\frac{\sqrt{ab}}{L}.$$
Therefore,
\begin{align*}
  \delta_{a}+C_{a,b,k}^{\alpha,\omega}\theta_{a,b}&\leq \frac{L-s}{L+s}+\max\left\{ \frac{s}{\sqrt{ab}}, \sqrt{\frac{s}{a}}\right\}\cdot(1+\frac{L-s}{L+s})\frac{\sqrt{ab}}{L} \\
  &=\frac{L-s}{L+s}+\frac{s}{\sqrt{ab}}(1+\frac{L-s}{L+s})\frac{\sqrt{ab}}{L} \\
  &=1.
\end{align*}
$\mathrm{(ii)}$ When $b>s,$ without loss of generality, we can suppose that $u$ and $v$
are nonzero. If $u=0$ or $v=0$,  clearly $\langle Au, Av\rangle=0\leq C \|u\|_{2}\|v\|_{2}$  holds for all $C>0$.
We normalize $u$ and $v$ such that $\|u\|_{2}=\|v\|_{2}=1$.
Because $u$ is $a-$sparse and $v$ is $b-$sparse, and $u, v$ have disjoint supports,
we conclude
$$|l_{1}|=|\langle \xi_{1}, u\rangle|\leq \sqrt{\frac{a}{L}}\|u\|_{2}=\sqrt{\frac{a}{L}}=\sqrt{\frac{a}{s+a}},$$
and
\begin{align*}
  \left| d_{1}\pm \sqrt{\frac{a}{s}}l_{1}\right| & =  \left| \left\langle \xi_{1}, v\pm \sqrt{\frac{a}{s}}u\right\rangle\right|
\leq\left\| v\pm \sqrt{\frac{a}{s}}u\right\|_{2} \\
&=\sqrt{ \|v\|_{2}^{2}+\frac{a}{s}\|u\|_{2}^{2}}=\sqrt{\frac{s+a}{s}}.
\end{align*}
In view of $|l_{1}|\leq \sqrt{\frac{a}{a+s}}$ and $1\leq a\leq s$,
\begin{align*}
  \frac{1}{1+\frac{L-s}{L+s}}|\langle Au, Av\rangle|& =|\sum\limits_{i=2}^{N}l_{i}d_{i}|=|-l_{1}d_{1}| \\
   & =\left( \max\left\{ \left| d_{1}+\sqrt{\frac{a}{s}}l_{1}\right|, \left| d_{1}-\sqrt{\frac{a}{s}}l_{1}\right| \right\}
   -\left|\sqrt{\frac{a}{s}}l_{1}\right|  \right)\cdot|l_{1}|  \\
   &\leq |l_{1}|\left(\sqrt{\frac{s+a}{s}}-\sqrt{\frac{a}{s}}|l_{1}| \right)\\
   &=-\sqrt{\frac{a}{s}} \left(|l_{1}|^{2}-\sqrt{\frac{s+a}{a}}|l_{1}| \right)\\
    &=-\sqrt{\frac{a}{s}} \left( |l_{1}|-\frac{1}{2}\sqrt{\frac{s+a}{a}}\right)^{2}+\frac{s+a}{4\sqrt{as}}\\
    &\leq -\sqrt{\frac{a}{s}} \left( \sqrt{\frac{a}{s+a}}-\frac{1}{2}\sqrt{\frac{s+a}{a}}\right)^{2}+\frac{s+a}{4\sqrt{as}}\\
    &=\frac{\sqrt{as}}{s+a}=\frac{\sqrt{as}}{L},
    \end{align*}
which implies
$$\theta_{a,b}\leq (1+\frac{L-s}{L+s})\frac{\sqrt{as}}{L}.$$
Hence,
\begin{align*}
  \delta_{a}+C_{a,b,k}^{\alpha,\omega}\theta_{a,b}&\leq \frac{L-s}{L+s}+\max\left\{ \frac{s}{\sqrt{ab}}, \sqrt{\frac{s}{a}}\right\}\cdot(1+\frac{L-s}{L+s})\frac{\sqrt{as}}{L} \\
  &=\frac{L-s}{L+s}+ \sqrt{\frac{s}{a}}(1+\frac{L-s}{L+s})\frac{\sqrt{as}}{L} \\
  &=1.
\end{align*}
In a word, $\delta_{a}+C_{a,b,k}^{\alpha,\omega}\theta_{a,b}\leq 1$ has been proved.

Next, we define
\begin{align*}
  \eta&=(\overbrace{1,\ldots,1}^{k-\alpha\rho k}, \overbrace{0,\ldots,0}^{\rho k}, \overbrace{1,\ldots,1}^{\alpha\rho k},0,\ldots,0)\in\mathbb{R}^{N}, \\
 \gamma&=(\underbrace{0,\ldots, 0}_{k-\alpha\rho k}, \underbrace{-1,\ldots,-1,}_{\rho k}
  \underbrace{0,\ldots, 0}_{\alpha\rho k},\underbrace{-1,\ldots,-1}_{L-k-\rho k}, 0,\ldots,0)\in\mathbb{R}^{N},  \quad\mathrm{if}~ L-k>\rho k,\\
   \mathrm{or}\quad
   ~\gamma&=(\underbrace{0,\ldots, 0}_{k-\alpha\rho k}, \underbrace{\overbrace{-1,\ldots,-1}^{L-k},0,\ldots,0}_{\rho k},
  \underbrace{0,\ldots, 0}_{\alpha\rho k}, 0,\ldots,0)\in\mathbb{R}^{N}, \quad \mathrm{if} ~L-k\leq\rho k.
  \end{align*}
From $1\leq a\leq s\leq k$ and $L=a+s$, we have $L-k\leq k$. Hence
$\eta$ and $ \gamma$ are $k-$sparse. Moreover, $\|\eta\|_{1,\mathrm{w}}=k$, $\|\gamma\|_{1, \mathrm{w}}\leq L-k\leq k$.
Note that  $\|\gamma\|_{1, \mathrm{w}} \leq \|\eta\|_{1,\mathrm{w}}$ and $\xi_{1}=\frac{1}{\sqrt{L}}(\eta-\gamma)$.
Since $A\xi_{1}=0$, we obtain $A\eta=A\gamma$.

(i) $\|\gamma\|_{1, \mathrm{w}} < \|\eta\|_{1,\mathrm{w}}$.

In the noiseless case $y=A\eta$, if weighted $l_{1}$ minimization
method (\ref{f2}) can exactly recover $\eta$, namely,
$\widehat{\eta}=\eta$. Clearly, $\|\widehat{\eta}\|_{1,
\mathrm{w}}=\|\eta\|_{1, \mathrm{w}}$. It contradicts that
$\|\gamma\|_{1, \mathrm{w}} < \|\eta\|_{1,\mathrm{w}}$.

In the noise case $y=A\eta+z$, suppose weighted $l_{1}$ minimization method (\ref{f2})
can stable recover $\eta$ with constraint $\mathcal{B}$,
i.e., $\lim\limits_{z\rightarrow 0}\widehat{\eta}=\eta$.
Due to $y-A(\widehat{\eta}-\eta+\gamma)=y-A\widehat{\eta} \in \mathcal{B}$ and the definition of
$\widehat{\eta}$, it follows immediately that $\|\widehat{\eta}\|_{1, \mathrm{w}}\leq \|\widehat{\eta}-\eta+\gamma\|_{1, \mathrm{w}}$. Thus, we have $\|\eta\|_{1, \mathrm{w}}\leq \|\gamma\|_{1, \mathrm{w}}$ as $z\rightarrow 0$. It contradicts that $\|\gamma\|_{1, \mathrm{w}} < \|\eta\|_{1,\mathrm{w}}$.

(ii) $\|\gamma\|_{1, \mathrm{w}} = \|\eta\|_{1,\mathrm{w}}$. The
weighted $l_{1}$ method (\ref{f2}) does not distinguish $k-$sparse
signals $\eta$ and $\gamma$ based $y$ and $A$.

Hence the weighted $l_{1}$ method (\ref{f2}) does not exactly and
stably recover the $k-$sparse signal $\eta$ based on $A$ and $y$.
Combining Theorem \ref{t3}, we have
$\delta_{a}+C_{a,b,k}^{\alpha,\omega}\theta_{a,b}=1$. This completes
the proof of the theorem.

\end{proof}

\begin{proof}[Proof of Proposition \ref{p1}]
For (1) and (2), when $\omega=1$ or $\alpha=\frac{1}{2}$, by simple
calculation, we have $s=2k-a, d=k$. Then, it is easy to imply (1)
and (2) by comparing Theorem \ref{t3} with Theorem \ref{t1} and
Theorem \ref{t2}.

(3) Let $0\leq\omega<1$. If $\alpha>\frac{1}{2}$, by means of the definition of $s$ in (\ref{g6})
and $d$ in (\ref{g7}), it follows immediately that $s<2k-a, d=k$.

When $b\leq s$, $C_{a,b,k}^{\alpha,\omega}=\frac{s}{\sqrt{ab}}<\frac{2k-a}{\sqrt{ab}}=C_{a,b,k}$.

When $s<b\leq 2k-a$, $C_{a,b,k}^{\alpha,\omega}=\sqrt{\frac{s}{a}}<\frac{2k-a}{\sqrt{ab}}=C_{a,b,k}$.

When $b\geq 2k-a$, $C_{a,b,k}^{\alpha,\omega}=\sqrt{\frac{s}{a}}<\sqrt{\frac{2k-a}{a}}=C_{a,b,k}$.

For any positive integers $a$ and $ b$ with $1\leq a\leq k$, in short, we obtain
$C_{a,b,k}^{\alpha,\omega}<C_{a,b,k}$,
which implies $\delta_{a}+C_{a,b,k}^{\alpha,\omega}\theta_{a,b}<\delta_{a}+C_{a,b,k}\theta_{a,b}.$
Thus， the condition $\delta_{a}+C_{a,b,k}^{\alpha,\omega}\theta_{a,b}<1$ in (\ref{g23}) is weaker than
$\delta_{a}+C_{a,b,k}\theta_{a,b}<1$ in (\ref{g1}) and $D_{0}=\frac{\sqrt{2(1+\delta_{a})k/a}}{1-\delta_{a}-C_{a,b,k}^{\alpha,\omega}\theta_{a,b}}
<\frac{\sqrt{2(1+\delta_{a})k/a}}{1-\delta_{a}-C_{a,b,k}\theta_{a,b}}=C_{0}$, $D'_{0}=\frac{\sqrt{2k}}{1-\delta_{a}-C_{a,b,k}^{\alpha,\omega}\theta_{a,b}}<\frac{\sqrt{2k}}{1-\delta_{a}-C_{a,b,k}\theta_{a,b}}=C'_{0}$, which implies (3).

(4) Assume $0\leq\omega<1$. If $\alpha>\frac{1}{2}$ and $b\leq s$,  we have $d=k$ and $\delta_{a}+C_{a,b,k}^{\alpha,\omega}\theta_{a,b}<\delta_{a}+C_{a,b,k}\theta_{a,b}.$
Combining the definition of $C_{1}$ and $D_{1}$, obviously,  $D_{1}=\frac{\sqrt{2k}\frac{1}{\sqrt{ab}}\theta_{a,b}}
{1-\delta_{a}-C_{a,b,k}^{\alpha,\omega}\theta_{a,b}}+\frac{1}{\sqrt{k}}
<\frac{\sqrt{2k}\frac{1}{\sqrt{ab}}\theta_{a,b}}
{1-\delta_{a}-C_{a,b,k}\theta_{a,b}}+\frac{1}{\sqrt{k}}=C_{1}.$

(5) Since $\alpha>\frac{1}{2}$ and $s<b\leq 2k-a$,
$C_{a,b,k}^{\alpha,\omega}=\sqrt{\frac{s}{a}},
C_{a,b,k}=\frac{2k-a}{\sqrt{ab}}$. Thus, to prove $D_{1}<C_{1}$, we
just need to prove $\frac{\sqrt{2k}\sqrt{\frac{1}{a}}\theta_{a,b}}
{(1-\delta_{a}-\sqrt{\frac{s}{a}}\theta_{a,b})\sqrt{s}}+\frac{1}{\sqrt{k}}
<\frac{\sqrt{2k}\sqrt{\frac{1}{ab}}\theta_{a,b}}
{1-\delta_{a}-\frac{2k-a}{\sqrt{ab}}\theta_{a,b}}+\frac{1}{\sqrt{k}}$.
It is equal to prove that $1-\delta_{a}
<\frac{2k-a-s}{\sqrt{a}(\sqrt{b}-\sqrt{s})}\theta_{a,b}$, namely,
$1-\delta_{a}-C_{a,b,k}^{\alpha,\omega}\theta_{a,b}
<\frac{2k-a-\sqrt{bs}}{\sqrt{a}(\sqrt{b}-\sqrt{s})}\theta_{a,b}$.

(6) Due to $\alpha>\frac{1}{2}$ and $b>2k-a$, we have
$C_{a,b,k}^{\alpha,\omega}=\sqrt{\frac{s}{a}}$ and
$C_{a,b,k}=\sqrt{\frac{2k-a}{a}}$. To show $D_{1}<C_{1}$ is equal to
prove that $\frac{\sqrt{2k}\sqrt{\frac{1}{a}}\theta_{a,b}}
{(1-\delta_{a}-\sqrt{\frac{s}{a}}\theta_{a,b})\sqrt{s}}+\frac{1}{\sqrt{k}}
<\frac{\sqrt{2k}\sqrt{\frac{1}{a}}\theta_{a,b}}
{(1-\delta_{a}-\sqrt{\frac{2k-a}{a}}\theta_{a,b})\sqrt{2k-a}}+\frac{1}{\sqrt{k}}$.
It suffices to prove
$1-\delta_{a}<\frac{\sqrt{2k-a}+\sqrt{s}}{\sqrt{a}}\theta_{a,b}$,
i.e., $1-\delta_{a}-C_{a,b,k}^{\alpha,\omega}\theta_{a,b}
<\sqrt{\frac{2k-a}{a}}\theta_{a,b}$.
\end{proof}

\section*{Acknowledgments}
This work was supported by the NSF of China (Nos.11271050, 11371183)
and Beijing Center for Mathematics and Information Interdisciplinary
Sciences (BCMIIS).

\end{document}